\documentclass[runningheads, envcountsame, a4paper]{llncs} 

\usepackage[T1]{fontenc}
\usepackage[utf8]{inputenc}
\usepackage{lmodern}
\usepackage{comment}

\usepackage[english]{babel}
\usepackage{microtype}

\usepackage{graphicx}
\usepackage{xcolor}

\usepackage{amsmath}
\usepackage{amsfonts,amssymb}

\newcommand{\N}{\mathbb{N}}

\newcommand{\Q}{\mathbb{Q}}
\newcommand{\R}{\mathbb{R}}

\newcommand{\BAfL}{\mathsf{BAfL}}

\newcommand{\AfL}{\mathsf{AfL}}

\newcommand{\SL}{\mathsf{SL}}

\newcommand{\REG}{\mathsf{REG}}

\usepackage{enumitem}
\setlist[itemize]{topsep=0pt,itemsep=-1ex,partopsep=1ex,parsep=1ex}

\newcommand{\EquationSpaceGlobal}{6pt}
\AtBeginDocument{
  \setlength{\belowdisplayskip}{\EquationSpaceGlobal} 
  \setlength{\belowdisplayshortskip}{\EquationSpaceGlobal}
  \setlength{\abovedisplayskip}{\EquationSpaceGlobal} 
  \setlength{\abovedisplayshortskip}{\EquationSpaceGlobal}
}

\newtheorem{coro}[theorem]{Corollary}

\newcommand{\A}{\mathcal{A}}
\newcommand{\B}{\mathcal{B}}
\newcommand{\C}{\mathcal{C}}
\newcommand{\D}{\mathcal{D}}

\begin{document}

\mainmatter  

\title{On the computational power of affine automata}

\titlerunning{On the computational power of affine automata}

%
%
\author{Mika Hirvensalo\inst{1}\fnmsep\inst{4}\fnmsep\thanks{Partially supported by V\"ais\"al\"a Foundation}%
\and Etienne Moutot\inst{2}\fnmsep\thanks{Partially supported by TUCS COM${}^3$-project and ANR project CoCoGro (ANR-16-CE40-0005)}\and Abuzer Yakary\i lmaz\inst{3}\fnmsep\thanks{Partially supported by TUCS COM${}^3$-project and ERC Advanced Grant MQC}
}
\authorrunning{M. Hirvensalo, E. Moutot, and A. Yakary\i lmaz}

\institute{Department of Mathematics and Statistics, University of Turku, FI-20014 Turku, Finland
\email{mikhirve@utu.fi}
\and { 
LIP, ENS de Lyon – CNRS – INRIA – UCBL – Universit\'e de Lyon
, \'Ecole Normale Sup\'erieure de Lyon, 69007 Lyon, France \email{etienne.moutot@ens-lyon.org}}
\and {Faculty of Computing, University of Latvia, Raina bulv. 19, R\={\i}ga, LV-1586, Latvia
\email{abuzer@lu.lv}
}
\and Turku Centre for Computer Science (TUCS)}

%
%

\toctitle{On the computational power of affine automata}
\tocauthor{Mika Hirvensalo, Etienne Moutot, and Abuzer Yakary\i lmaz}
\maketitle

\begin{abstract} We investigate the computational power of affine automata (AfAs) introduced in \cite{Yaka16}. In particular, we present a simpler proof for how to change the cutpoint for any affine language and a method how to reduce error in bounded error case. 
 Moreover, we address to the question of \cite{Yaka16} by showing that any affine language can be recognized by an AfA with certain limitation on the entries of affine states and transition matrices. Lastly, we present the first languages shown to be not recognized by AfAs with bounded-error. 
\keywords{non-classical models of automata, affine automata, cutpoint languages, bounded error, compact sets, error reduction}
\end{abstract}

\section{Introduction}
Finite automata are interesting computational models because of their simplicity, compared to more complex models like pushdown automata or Turing machines.
They also represent a very concrete restriction on computation: they only have a finite memory.
A lot of different automata models have been studied during the years, such as deterministic \cite{Sipser}, probabilistic \cite{Paz} and quantum \cite{Ambainis2015} ones.
All these models share two common features: the state vector set is compact and the acceptance function can be interpreted as linear. The linearity is desirable because of mathematical simplicity, but on the other hand, it may represent a limitation on the computational power. 

Jeandel \cite{Jeandel2007} demonstrated that in the bounded-error acceptance model the finite automata with compact state set accept only regular languages. Hence the compactness property of the state set may be one very important limiting the computational power, but since most known models have a compact state set, it remains open how much the compactness of the state set actually contributes.

Recently, A. Díaz-Caro and A. Yakary\i lmaz introduced a new model, called {\em affine automata} \cite{Yaka16}, also investigated in \cite{Yaka162} and \cite{Yaka163}. It is a purely theoretical model, which means that it cannot be implemented by a physical device like quantum automata. But it allows us to investigate on the power of interference caused by negatives amplitudes in the computation, like in the quantum case. Moreover, this model allows us to study the effect of state set compactness, since unlike quantum automata, affine ones have an unbounded state set. In addition, the final operation corresponding to quantum measurement cannot be interpreted as linear, but it is analogous to renormalization in Kondacs-Watrous \cite{KW97} and Latvian \cite{Ambainis06} quantum automata models.

In this paper, we present some stability results (Section \ref{section:stability}): we show how to obtain a new AfA from two AfAs by tensoring and direct sum. Then, we present a simpler proof for how to change the cutpoint for any affine language and an error reduction method in bounded error case. 

Any entry of an affine state or a transition matrix can be arbitrarily away from zero. Here, by addressing to the question of \cite{Yaka16}, we show that (Section \ref{section:equivalent}) any affine language can be recognized by an AfA with the restriction that all the entries of transition matrices are in the interval $ [-1,1] $. We also show that by an additional state we can guarantee that any AfA can start its computation from the first deterministic state.

Finally, we present (Section \ref{sec:non-affine}) the first languages shown not to be recognized by any bounded-error AfA. 

\newcommand{\estr}{\varepsilon}

\section{Preliminaries}
\label{section:def}

We denote the input alphabet $ \Sigma $ and the empty string $ \estr  $.

Probabilistic automata are a generalization of deterministic finite automata that can make random choices \cite{Rabin}. 
Formally, 
a {\em probabilistic finite automaton} (PFA) $P$ is a 5-tuple
$
P = (E, \Sigma, \{M_x ~|~ x\in\Sigma\}, e_s, E_a),
$ 
where $E = \{e_1, \dots, e_k\}$ is the finite set of states of $P$,
$\{M_x ~|~ x\in\Sigma \}$ is the set of stochastic transition matrices (all their coefficients are real numbers in $[0,1]$ and their columns sums up to $1$), $ v_0$ is the initial probabilistic state (the probability distribution on the states), and $E_a\subseteq E$ is the set of accepting states.
The computation starts in $ v_0 $, and then the given input, say $w = w_1 \cdots w_n \in \Sigma^* $ for some $n > 0$, is read once from left to right symbol by symbol. For each symbol the corresponding transition matrix is applied:
$
\vec{v}_f =M_w \vec{v}_0 = M_{w_n} \cdots M_{w_1} \vec{v}_0.
$
Remark that if $ w = \varepsilon $, $ v_f = v_0 $.
The {\em accepting probability}  of $P$ on $w$ is given by
\begin{equation}\label{finalvalue01}
f_P(w) = \vec{p}M_w \vec{v}_0,
\end{equation}
  where $\vec{p} = 
  \begin{pmatrix}
    \delta_1 & \cdots & \delta_k
  \end{pmatrix} $ 
  and
  $\delta_i = 1$ if $e_i\in E_a$ and $0$ if $e_i\notin E_a$.

Affine automata are a generalization of PFAs allowing negative transition values. Only allowing negative values in the transition matrices does not add any power (generalized probabilistic automata are equivalent to usual ones \cite{Turakainen1969}), but affine automata introduces also a non-linear behaviour. The automaton acts like usual generalized probabilistic automaton until the last operation, a non-linear operation called \textit{weighting}. 

A vector $\vec{v}\in\R^n$ is an affine vector if and only if its coordinates sums up to $1$. A matrix $M$ is an affine matrix if and only if all its columns are affine vectors. 
Remark that
if $M$ and $N$ are affine matrices , then $MN$ is also an affine matrix. In particular, if $\vec{v}$ is an affine vector, then $M\vec{v}$ is also an affine vector.

Formally,  an {\em affine finite automaton} (AfA) $A$ is a 5-tuple
  \[
	A = (E, \Sigma, \{M_x ~|~ x\in\Sigma\}, \vec{v}_0, E_a)
	\] 
where all components exactly the same as for probabilistic automata by replacing stochastic property with affine one in the initial state and transition matrices.

As in PFAs, after reading a word $w=w_1 \cdots w_n$, the final state of $A$ is $\vec{v}_f = M_w \vec{v}_0 $ like in the probabilistic case, but the function $f_A:\Sigma^*\to[0,1]$ computed by $A$ is defined as
\begin{equation}
	\label{finalvalue02}
	f_A(w) =  \frac{\sum_{e_i\in E_a} |(v_f)_i|}{\sum_{e_i\in E} |(v_f)_i|},
\end{equation}
and referred as the {\em accepting value} of $A$ on $w$. Similar to projective measurements, we can rewrite Eq. (\ref{finalvalue02}) as given below. 
First, we define a projection matrix based on $ E_a $: 
 {\footnotesize $ 
  P_A = P=  
  \begin{pmatrix}
    \delta_1 &          &        & \\
             & \delta_2 &        & \\
             &          & \ddots & \\
             &          &        & \delta_n \\
  \end{pmatrix} ,
  $}
  $
  \text{ where }
  \delta_i = 
  \left\{
  \begin{matrix}
    1 \text{ if } e_i\in E_a\\
    0 \text{ otherwise}
  \end{matrix}
  \right.
  .$ 
\\
Then, we can denote $ f_A(\cdot) $ as
\begin{equation}\label{finalvalue03}
f_A(w) = \frac{|PM_w\vec{v}_0|}{|M_w \vec{v}_0|}  .
\end{equation}
Notice that the final value for PFA $ P $ (\ref{finalvalue01}) is defined as matrix product $\vec{v}_f\mapsto\vec{p}.\vec{v}_f$, which is a linear operation on $\vec{v}_f$. On the other hand, computing final value from $\vec{v}_f$ as in (\ref{finalvalue03}) involves nonlinear operations $\displaystyle\vec{v}_f\mapsto\frac{|P\vec{v}_f|}{|\vec{v}_f|}$ due to absolute value and normalization of affine states having length greater than 1.



Given a function $f:\Sigma^*\to[0,1]$ computed by an automaton (stochastic or affine), there are different ways of defining the language of an automaton. The natural one is as follows:
  \label{def:SL}
  A language $L\subseteq\Sigma^*$ is recognized by an automaton $A$ with cutpoint $\lambda$ if and only if 
  \[ L = \{ w\in\Sigma^* ~|~ f_A(w) > \lambda \} .\]
These languages are called cutpoint languages.	
In the case of probabilistic (resp. affine automata), the set of cut-point languages are called {\em stochastic languages} (resp. {\em affine languages}) and denoted by $\SL$ (resp. $\AfL$).
	

A stronger condition is to impose that accepted and rejected words are separated by a gap: the cutpoint is said to be isolated: 
A language $L$ is recognized by an automaton $A$ with {\em isolated cutpoint} $\lambda$ if and only if there exist $\delta>0$ such that $\forall w\in L, f_A(w) \geq \lambda+\delta$, and
$\forall w\notin L, f_A(w) \leq \lambda-\delta$.

As we shall see, for affine automata it is always possible to shift the cutpoint $\lambda\in(0,1)$
to $\lambda=\frac12$, and hence this notion of isolated cutpoint becomes equivalent to the bounded error recognition: Language $L\subseteq \Sigma^*$ is said to be recognized by an automaton $A$ with {\em bounded error} if and only if there exists $\varepsilon>0$ such that $\forall w\in L, f_A(w) \geq 1-\varepsilon$, and
$\forall w\notin L, f_A(w) \leq \varepsilon$.

The set of languages recognized with {\em bounded error} (or isolated cutpoint) affine automata is denoted by $\BAfL$.

A classical result by Rabin \cite{Rabin} shows that isolated cutpoint stochastic languages are regular (denoted $\REG$). Rabin's proof essentially relies on two facts: 1) the function mapping the final vector into $[0,1]$ is a contraction, and 2) the state vector set is bounded.

By modifying Rabin's proof, it is possible to show that also many quantum variants of stochastic automata obey the same principle \cite{BQAL} bounded-error property implies the regularity of the accepted languages. In fact, E. Jeandel generalized Rabin's proof by demonstrating that the compactness of the state vector set together with the continuity of the final function are sufficient to guarantee the regularity of the accepted language if the cutpoint is isolated \cite{Jeandel2007}.

  


In the affine case however, the vector states do not lie in a compact set, we cannot prove that $\BAfL = \REG$ like in the probabilistic (or even quantum) case \cite{Jeandel2007}. In fact, it is even the contrary: $\REG \subsetneq \BAfL$ \cite{Yaka16}.

We close this section by three basic facts.
The following three operations on the state sets will be useful, when constructing new automata from the existing ones:
  \begin{itemize}
    \item $\overline{E} = \{e_i ~|~ e_i \notin E\}$ the complement of $E$,
    \item $E_a \times E_b = \{ (e_{i}, e_j) ~|~ e_i \in E_a, e_j \in E_b \}$ the Cartesian product of $E_a$ and $E_b$,
    \item $E_a \cup E_b = \{ e_{i} ~|~ e_i \in E_a \text{ or } e_i \in E_b \}$ the union of $E_a$ and $E_b$.
  \end{itemize}
The following lemma shows how to formulate the above operations by using the formalism of projection matrices.
 
\begin{lemma}
  \label{lemma:projection_operations}
  Let $E$ be the set of all states, $E_a, E_b \subseteq E$ and $P_a$, $P_b$ be the projections associated to them. Then
  \begin{itemize}
    \item $P$ is the projection associated to the complement $\overline{E_a}$ if and only if $P = I - P_a$, and,
    \item $P$ is the projection associated to $E_a \times E_b$ if and only if $P = P_a \otimes P_b$.
  \end{itemize}
\end{lemma}
\begin{lemma}\label{lem02}
  Let $E$ be the set of all states, $E_a, E_b \subseteq E$, such that $E_a \cap E_b = \emptyset$. Let $P_a$ and $P_b$ be the projections associated to them. Then,
  \begin{itemize}
    \item $P$ is the projection associated to $E_a \cup E_b$ if and only if $P = P_a + P_b$, and,
    \item for any matrix $M$ and vector $\vec{v}$, $|PM\vec{v}| = |P_aM\vec{v}|+|P_bM\vec{v}|$.
  \end{itemize}
\end{lemma}
\begin{lemma}
  \label{lemma:affine_prod}
  If $A$ and $B$ are affine matrices, then $A\otimes B$ is also affine. Moreover, $|A\otimes B| = |A||B|$.
\end{lemma}

\section{Stability Results}
\label{section:stability}
The main results of this section are stability results.
The first are about the functions of affine automata. They provide a way to prove an error reduction theorem. We then use this theorem to show the stability of bounded-error affine languages under intersection and union.


\begin{proposition}
\label{th:affine_product}
  Let $f$, $g$ be functions computed by affine automata, then there exists an affine automaton $\C$ such that $f_\C = f\times g$.
\end{proposition}

\begin{proof}
  The proof is the same as the stochastic case and essentially relies on the property of tensor product of Lemma \ref{lemma:affine_prod}.
\qed
\end{proof}
It is easy to design a 2-state PFA $ P $ such that $ f_P : \Sigma^* \rightarrow \alpha $ for $ \alpha \in [0,1] $. Thus:
\begin{coro}
  Let $f$ be a function computed by an AfA and $\alpha\in[0,1]$, then there exists an AfA $\C$ such that $f_\C = \alpha f$.
\end{coro}

\begin{proposition}
	\label{prop:direct-sum}
  Let $f$, $g$ be functions computed by some AfAs and $\alpha, \beta \geq 0$ such that $\alpha+\beta=1$, then there exists an AfA $\C$ such that $f_\C = \alpha f+ \beta g$.
\end{proposition}

\begin{proof}
	Let $\A = (E^A, \Sigma, \{A_x\}, \vec{v}_0^A, E_a^A)$ and $\B = (E^B, \Sigma, \{B_x\}, \vec{v}_0^B, E_a^B)$ two automata such that $f = f_\A$ and $g=f_\B$. 
    The idea here is to make two copies of $\A\otimes \B$ working in parallel, one having the final states of $\A$, the other the final states of $\B$. We define $\C = (E^C, \Sigma, \{C_x\}, \vec{v}_0^C, E_a^C)$ by:
    {\footnotesize
  \[C_x 
  = 
  \left(
  \begin{array}{c|c}
      \\
    \A_x \otimes B_x & 0  \\
      \\
    \hline
      \\
      0 & A_x \otimes B_x\\
      \\
  \end{array} 
  \right)
   ,
   \vec{v}_0^C = 
  \left(
  \begin{array}{c}
    \\
    \alpha (\vec{v}_0^A\otimes \vec{v}_0^B)\\
    \\
    \hline
    \\
    \beta (\vec{v}_0^A\otimes \vec{v}_0^B)\\
    \\
  \end{array} 
  \right)
   , 
   P^C = 
   \left(
  \begin{array}{c|c}
      \\
    P^A \otimes I_n & 0  \\
      \\
    \hline
      \\
      0 & I_k \otimes P^B\\
      \\
  \end{array} 
  \right)
   ,\] }
   with $P^A$, $P^B$ and $P^C$ be the projections on $E^A_a$, $E^B_a$ and $E^C_a$.
   Thus,
   \begin{align*}
     f_\C(w) 
     &= \frac{\alpha |(P^A \otimes I_n) (A_x \otimes B_x) (\vec{v}_0^A\otimes \vec{v}_0^B)| + \beta |(I_k \otimes P^B) (A_x \otimes B_x) (\vec{v}_0^A\otimes \vec{v}_0^B)|} 
       {(\alpha +\beta )|(A_x \otimes B_x) (\vec{v}_0^A\otimes \vec{v}_0^B)|} \\
    &= \alpha \frac{|P^A A_w\vec{v}_0^A|}{|A_w\vec{v}_0^A |} + \beta \frac{|P^B B_w \vec{v}_0^B|}{|B_w \vec{v}_0^B|} = \alpha f(w) + \beta g(w)
   .\end{align*} 
   \qed
\end{proof}

The first consequence of these stability results is a really short proof for shifting the cutpoint of an affine automaton. Although the construction in \cite{Yaka16} gives a much more compact automata in term of number of states, our construction is simpler, and does not require as many specific cases.

\begin{proposition}
  \label{th:cutpoint_shift}
  Let $\A$ be an affine automaton and $\lambda_1, \lambda_2 \in [0,1]$. There exists an affine automaton $\B$ such that 
  \begin{itemize}
    \item $f_\A(w) > \lambda_1 \Leftrightarrow f_\B(w) > \lambda_2$ and
    \item $f_\A(w) = \lambda_1 \Leftrightarrow f_\B(w) = \lambda_2$.
  \end{itemize}
\end{proposition}

\begin{proof}
  First we suppose $\lambda_1 \neq 1$. Let $\B$ the automaton such that $f_\B = \alpha f_\A + (1-\alpha)1$, with $\alpha = \frac{1 - \lambda_2}{1 - \lambda_1}$.
  Then $f_\A > \lambda_1 \Rightarrow f_\B > \frac{(1 - \lambda_2)\lambda_1 + \lambda_2 - \lambda_1}{1-\lambda_1} = \lambda_2$. And one has the same with $=$ or $<$.
  
  For $\lambda_1 = 1$ it is even simpler, one has just to ``resize'' the function by taking $\B$ such that $f_\B = \lambda_2 f_\A$. And then, $f_\A = 1 \Rightarrow f_\B = \lambda_2$, and same for $<$.
\qed
\end{proof}

Using the same kind of construction we can prove that bounded-error mode, it is always possible to reduce the error. Reducing the error means increasing the gap between accepted and rejected words. The error probability could even be made as close to zero as one wants.

\begin{lemma}
  \label{lemma:poly_func}
  Let $f$ be a function computed by affine automaton, then there exists an affine automaton $\B$ such that $f_\B = f^2(3-2f)$.
\end{lemma}

\begin{proof}
  Let $\A = (E, \Sigma, \{A_x\}, \vec{v}_0, E_a)$ such that $f = f_\A$. The automaton $\B$ will run 3 copies of $\A$ in parallel, and its final states are made to accept if 2 or 3 copies of $\A$ accept and reject otherwise (i.e. taking the majority answer). Formally, $\B = (E\otimes E \otimes E, \Sigma, \{B_x\}, \vec{v}'_0, E_a')$ with
  \[ B_x = A_x \otimes A_x \otimes A_x ,\]
  \[ \vec{v}'_0 = \vec{v}_0 \otimes \vec{v}_0 \otimes \vec{v}_0 ,\]
  \[ E'_a =  \left(E_a \times E_a \times E_a\right) \cup \left(\overline{E_a} \times E_a \times E_a\right) \cup \left(E_a \times \overline{E_a} \times E_a\right) \cup \left(E_a \times E_a \times \overline{E_a}\right) .\]
  Note that the four sets in parenthesis are all pairwise disjoints.
  Let $P$ and $P'$ be the projections associated to $E_a$ and $E'_a$. Then,
  \[  P' = P \otimes P \otimes P + (I - P) \otimes P \otimes P + P \otimes (I - P) \otimes P + P \otimes P \otimes (I - P) .\]
  And by Lemma \ref{lemma:projection_operations},
  \begin{align*}
    f_\B(w) = \frac{|P'B_w\vec{v}'_0|}{|B_w\vec{v}'_0|} 
    &= \frac{|PA_w\vec{v}_0|^3 + 3 |PA_w\vec{v}_0| \left( |A_w\vec{v}_0|-|PA_w\vec{v}_0| \right)}{|A_w\vec{v}_0|^3} \\
    &= f(w)^3 + 3f(w)^2(1-f(w)) \\
    &= f(w)^2(3-2f(w))
    .
  \end{align*}
\qed
\end{proof}

\begin{proposition}[\textbf{error reduction}]
  \label{th:error_reduction}
  Let $L\in \BAfL$. There exists an affine automaton $\A$ such that:
  \begin{itemize}
    \item $\forall w\in L, f_\A(w) \geq \frac{3}{4}$
    \item $\forall w\notin L, f_\A(w) \leq \frac{1}{4}$
  \end{itemize}
\end{proposition}

\begin{proof} We do not detail the proof, but the idea is simple: mapping $x\to x^2(3-2x)$ has attracting points at $x=0$ and $x=1$. Iterating the mapping, any point $\ne \frac12$ will tend to $0$ or $1$.
\qed
\end{proof}

This technique could be applied to get any constant instead of $\frac{1}{4}$, to have an error bound as small as one wants.

This error reduction theorem also applies to probabilistic automata, but is not very interesting because in the probabilistic case it is known that bounded-error languages are exactly regular languages \cite{Jeandel2007}, and hence the error probability could always be 0. In our case, bounded-error languages are more complex than regular languages. But thanks to this error reduction, they are stable under union, intersection, and complement, just like regular languages.

\begin{proposition}
  Let $L_A, L_B\in \BAfL$. Then 
  \begin{itemize}
    \item $L_A \cup L_B \in \BAfL$,
    \item $L_A \cap L_B \in \BAfL$,
    \item $\overline{L_A}\in \BAfL$.
  \end{itemize}
\end{proposition}
\begin{proof}
  Let $\A$ and $\B$ be automata recognizing $L_A$ and $L_B$ with error bound $\varepsilon$ at most $\frac{1}{4}$ (thanks to Theorem \ref{th:error_reduction}). We define $\C$ and $\D$ such that $f_\C = \frac{1}{2}(f_\A + f_\B)$ and $f_\D = f_\A f_\B$. Let $w\in\Sigma^*$. We study the 4 possible options depending on the membership of $w$ to $L_A$ and $L_B$.
  \begin{itemize}
    \item $w\in L_A, w\in L_B$ {\footnotesize  (i.e. $w\in L_A \cup L_B, w\in L_A \cap L_B$)}   
    $\Rightarrow f_C \geq \frac{3}{4}$ and $f_\D \geq \frac{9}{16}$,
    \item $w\in L_A, w\notin L_B$ {\footnotesize (i.e. $w\in L_A \cup L_B, w\notin L_A \cap L_B$)}
    $\Rightarrow f_C \geq \frac{3}{8}$ and $f_\D \leq \frac{1}{4}$,
    \item $w\notin L_A, w\in L_B$ {\footnotesize (i.e. $w\in L_A \cup L_B, w\notin L_A \cap L_B$)}
    $\Rightarrow f_C \geq \frac{3}{8}$ and $f_\D \leq \frac{1}{4}$,
    \item $w\notin L_A, w\notin L_B$ {\footnotesize (i.e. $w\notin L_A \cup L_B, w\notin L_A \cap L_B$)}
    $\Rightarrow f_C \leq \frac{1}{4}$ and $f_\D \leq \frac{1}{16}$.
  \end{itemize}
  Because $\frac{3}{8} > \frac{1}{4}$ and $\frac{9}{16} > \frac{1}{4}$, $\C$ and $\D$ are deciding $L_A \cup L_B$ and $L_A \cap L_B$ with bounded error.
  
  For the complement one has just to make a copy of $\A$ with accepting states $\overline{E_a}$. The resulting function will be $1-f_\A$, leading to accept the rejected words of $\A$ and vice-versa.
\qed
\end{proof}

\section{Equivalent Forms of Affine Automata}
\label{section:equivalent}
General affine automata are hard to study because of the lack of structure of their transition matrices and state vectors.
We provide here some equivalent forms which have more restrictive properties.
These equivalent forms are useful not only because it provides simpler equivalent models but also because they provide a way understand the power of affine computation. 

The first result is that assuming the initial affine (probabilistic) state as the first deterministic state does not change the power of AfAs (PFAs).

\begin{proposition}
  \label{th:initial_vector}
  Let $\A$ be an affine automaton with $n$ states, there exist $\B$ with $n+1$ states with the initial state $(1,0,\ldots, 0)$ and such that $f_\A = f_\B$.
\end{proposition}
\begin{proof}
  Let $\mathcal A = (E, \Sigma, \{A_x\}, \vec{v}_0, E_a)$. Then,  
  $\B = (E\cup \{e'\}, \Sigma, \{B_x\}, \vec{v}'_0, E_a)$, with
  $\vec{v}'_0 = (1,0,\ldots, 0)^T$ and  
  $ B_x =  
  \left( 
    \begin{array}{c|ccc}
       0 & 0 & \cdots  & 0 \\
       \hline
        \\
        A_x \vec{v}_0 & & A_x\\
        \\ 
    \end{array} 
  \right)
  .
  $
  Thus we can deduce $f_\B = f_\A$ from
  $ B_w \vec{v}'_0 = B_{w_n}\ldots B_{w_2} B_{w_1} \vec{v}'_0 = 
  \left( 
    \begin{array}{c|ccc}
       0 & 0 & \cdots & 0 \\
       \hline
        \\
        A_w \vec{v}_0 & & A_w\\
        \\ 
    \end{array} 
  \right)
  \begin{pmatrix}
    1\\ 0 \\ \vdots \\ 0
  \end{pmatrix}
  =
  \begin{pmatrix}
    0\\  \\ A_w \vec{v}_0 \\ \\
  \end{pmatrix}
  . $
\qed
\end{proof}
Then we prove that one could also assume that all state vectors and transition matrices have coefficients only in $[-1, 1]$.
\begin{proposition}
  \label{th:bounded_cutpoint}
  Any language in $ \AfL $ can be recognized by a AfA $\B$ with cutpoint $ \frac{1}{2} $ such that each entry of affine states during the computation is always in $ [-1,1] $.
\end{proposition}

\begin{proof}
  Let $\A = (E=\{e_1, \ldots, e_k\}, \Sigma, \{A_x\}, v_0=(1,0,\ldots,0)^T, E_a)$ be an AfA such that $w\in L \Leftrightarrow f_\A(w) > \frac{1}{2}$, and $C = \max_{x,i,j} |(A_x)_{i,j}|$. Then, $\B$ is as follows:
  \[ \B = (E\cup \{e_{n+1}, e_{n+2}\}, \Sigma, \{B_x\}, \vec{v}'_0, E_a \cup \{e_{n+1}\}) ~~ \mbox{ with} \]
  \[ B_x = 
  \frac{1}{2kC} \left( 
    \begin{array}{ccc|cc}
       &     && 0 & 0             \\
       & 2A_x && \vdots & \vdots  \\
       &     && 0 & 0             \\
       \hline
       kC-1& \hdots & kC-1 & 2kC & 0 \\
       kC-1& \hdots & kC-1 & 0   & 2kC
    \end{array} 
  \right)  \mbox{ and } \vec{v}'_0 = 
  (1,0,\ldots,0 )^T 
  .\]
Then, with $w=w_1 \cdots w_n$, we can deduce that
	{\small
  \[ B_w = B_{w_n}\cdots B_{w_2}B_{w_1} = 
  \frac{1}{2(kC)^n}
  \left( 
  \begin{array}{ccc|cc}
       &     && 0  & 0      \\
       & 2 A_w && \vdots & \vdots  \\
       &     && 0 & 0       \\
       \hline
       (kC)^n-1 & \hdots & (kC)^n-1 & 2(kC)^n & 0 \\
       (kC)^n-1 & \hdots & (kC)^n-1 & 0       & 2(kC)^n
    \end{array} 
  \right)
  ,\] }
  which gives the final values of the states:
  \[\vec{v}'_f = B_w \vec{v}'_0 = 
  \frac{1}{(kC)^n}
  \begin{pmatrix}
   \vdots  \\
   \vec{v}_f     \\
   \vdots              \vspace{0.6em} \\
   \frac{(kC)^n -1}{2} \vspace{0.7em} \\
   \frac{(kC)^n -1}{2}
  \end{pmatrix}
  .\]
  Since $|(\vec{v}_f)_i| \leq k^{n-1}C^n$, it is clear that $|(\vec{v}'_f)_i| \leq [-1,1]$: the values of the states are bounded. 
  Now, one has
  \[ f_\B = \frac{|PA_w\vec{v}_0| + \frac{(kC)^n -1}{2}}{|A_w\vec{v}_0| + (kC)^n -1} ,\]
  and so,
  \begin{align*}
    w\in L \Leftrightarrow f_\A > \frac{1}{2} 
    &\Leftrightarrow |PA_w\vec{v}_0| > \frac{1}{2} |A_w\vec{v}_0| \\
    &\Leftrightarrow |PA_w\vec{v}_0| + \frac{(kC)^n -1}{2} > \frac{1}{2} \left( |A_w\vec{v}_0| + (kC)^n -1 \right) \\
    &\Leftrightarrow f_\B > \frac{1}{2}
  .
  \end{align*}
\qed
\end{proof}
\section{The first languages shown to be not in $\BAfL$}
\label{sec:non-affine}
This part is dedicated to prove that some languages are not recognizable by affine automata.
This is an adaptation of the proof of Turakainen \cite{Turakainen1981} for non-stochastic languages.
All the difficulty of exhibiting a non-affine language relies in the fact that a large majority of non-stochasticity proof are based on the linearity of the automaton, which is not the case in the affine case.
This proof however, is more based on some ``regularity'' induced by the matrix-based operations, and number theoretic properties of languages like $Prime$. 
Hence it was possible to adapt it for the affine case, where the only non-linear operation is the final projection.

 Let $L\subseteq a^*$ be a unary language. We call \textbf{lower density} of $L$ the limit
  \[ \underline{dens}(L) = \liminf_{n\rightarrow \infty} \frac{\left|\{a^k\in L ~|~ k\leq n\} \right|}{n+1} .\]
  Let $(\vec{x}_n)$ be a sequence of vectors in $\R^k$ and $I=[a_1, b_1) \times \cdots  \times [a_k, b_k)$ be an ``interval''. We define $C(I,n)$ as $C(I, n) = \left|\{ \vec{x}_i \mod 1 \in I ~|~ 1\leq i \leq n \}\right|$. 
  
  We say that \textbf{$(\vec{x}_n)$ is uniformly distributed mod 1} if and only if for any $I$ of such type,
  \[ \lim_{n\rightarrow \infty} \frac{C(I,n)}{n} = (b_1-a_1)\cdots(b_k-a_k) .\]


\begin{proposition}
  \label{th:non_affine}
  If $L\subseteq a^*$ satisfies the following conditions:
  \begin{enumerate}
    \item \label{cond:densL} \underline{dens}(L) = 0.
    \item \label{cond:ud}   For all $Q\in\N^*$, there exist $h\in\N$ and an infinite sequence $(n_i)\in\N^\N$ such that 
    $a^{h+n_iQ}\subseteq L$ and for any irrational number $\alpha$, the sequence $\left( (h+n_iQ)\alpha \right)_{i\in\N}$ is uniformly distributed mod 1.
  \end{enumerate}
  Then $L$ is non-affine ($L\notin \BAfL$).
\end{proposition}

\begin{proof}
Let's assume for contradiction that $L\in \BAfL$. Then there exists an affine automaton $A$ with $s$ states such that 
\[ f_A(a^n) = \frac{|PM^n\vec{v}|}{|M^n\vec{v}|} \]
and there exists $\varepsilon > 0$ such that 
\begin{itemize}
    \item $\forall w\in L, f_A(w) \geq 1-\varepsilon$,
    \item $\forall w\notin L, f_A(w) \leq \varepsilon$.
  \end{itemize}
Note that
\[ |M^n\vec{v}| 
=   \sum_{i=1}^s \left| (M^n \vec{v})_i \right| 
\geq \left|\sum_{i=1}^s (M^n \vec{v})_i \right|
= 1 
\text{ (triangle inequality).} \]
Hence the denominator of $f_A$ is never $0$, and so $f_A$ is continuous.

Using the Jordan decomposition $M=PJP^{-1}$, one has $M^n = PJ^nP^{-1}$. So the coordinates $\vec{v}_i$ of $M^n\vec{v}$ have the form
\begin{equation}
  \label{eq:vi_sum}
  \vec{v}_i = \sum_{k=1}^s p_{ik}(n)\lambda_k^n 
\end{equation}
where $\lambda_i$ are the eigenvalues of $M$ and $p_{ik}$ are polynomials of degree less than the degree of the corresponding eigenvalue.
Let $\lambda_i = |\lambda_i| e^{2i\pi\theta_i}$, we assume $|\lambda_1| = \cdots  = |\lambda_{s'}| > |\lambda_{s'+1}| \cdots$. Let $\lambda = |\lambda_1|$ be the largest module of all eigenvalues and $r$ be the maximum degree of all polynomials $p_{ik}$, where $k\leq s'$.
Then, one can use (\ref{eq:vi_sum}) to write 
\[ 
  |M^n\vec{v}| = \sum_{i\in E} |\vec{v}_i| = \lambda^n n^r \left( \sum_{i\in E} \left|\sum_{k=1}^{s'} a_{ik}e^{2i\pi n\theta_k} \right| + g_E(n)\right)
\]
where $a_{ik}$ is the coefficient of degree $r$ of $p_{ik}$ (note that one can have $a_{ik} = 0$ for some $a,k$), and $g_E$ a function such that $\lim_{n\rightarrow\infty} g_E(n) = 0$. Similarly, 
\[ 
  |PM^n\vec{v}| = \sum_{i\in E_a} |\vec{v}_i| = \lambda^n n^r \left( \sum_{i\in E_a} \left|\sum_{k=1}^{s'} a_{ik}e^{2i\pi n\theta_k} \right| + g_{E_a}(n)\right)
.\]

Now let $F(n) = f(a^n)$. Using the previous equations, one has
\begin{align*}
  F(n)
  &= \frac{|PM^n\vec{v}|}{|M^n\vec{v}|} \\
  &= \frac{\lambda^n n^r \left( \sum_{i\in E_a} \left|\sum_{k=1}^{s'} a_{ik}e^{2i\pi n\theta_k} \right| + g_{E_a}(n)\right)}
          {\lambda^n n^r \left( \sum_{i\in E} \left|\sum_{k=1}^{s'} a_{ik}e^{2i\pi n\theta_k} \right| + g_E(n)\right)} \\
  &= \frac{\sum_{i\in E_a} \left|\sum_{k=1}^{s'} a_{ik}e^{2i\pi n\theta_k} \right| + g_{E_a}(n)}
          {\sum_{i\in E} \left|\sum_{k=1}^{s'} a_{ik}e^{2i\pi n\theta_k} \right| + g_E(n)} .
\end{align*}

We define 
\[
  G(n) = \frac{\sum_{i\in E_a} \left|\sum_{k=1}^{s'} a_{ik}e^{2i\pi n\theta_k} \right|}
          {\sum_{i\in E} \left|\sum_{k=1}^{s'} a_{ik}e^{2i\pi n\theta_k} \right|} .
\]
As $\lim_{n\rightarrow\infty} g_{E_a}(n) = 0$ and $\lim_{n\rightarrow\infty} g_E(n) = 0$, one has $G(n) \sim F(n)$, and so,
\begin{equation}
  \label{eq:F-G=0}
  \lim_{n\rightarrow \infty} |F(n) - G(n)| = 0 .
\end{equation}

We define $A = \{ k ~|~ 1\leq k \leq s', \theta_k\notin\Q \}$ the indices of the ``first'' eigenvalue angles that are not rational.
Let $Q$, $h$ and the sequence $(n_i)$ be as in the statement. 
Using the periodic behaviour induced by rational angle of eigenvalues, and by taking a subsequence of the initial one, one can also assume that $(n_i)$ is such that
\[
  G(h + n_i Q) = \frac
                {\sum_{i\in E_a} \left|\sum_{k\in A} a_{ik}e^{2i\pi (h + n_i Q)\theta_k} + c \right|}
                {\sum_{i\in E}   \left|\sum_{k\in A} a_{ik}e^{2i\pi (h + n_i Q)\theta_k} + d \right|} 
\]
with $c$, $d$ some constants.

By assumption, for all $k\in A$, the sequence $\left((h+n_iQ)\theta_k\right)_i$ is uniformly distributed modulo 1. 
The consequence is that the values $e^{2i\pi (h + n_i Q)\theta_k}$ are dense in the unit circle. 
If for some $n$, $G(h+nQ) < \frac{1}{2}$, there exists $\varepsilon > 0$ such that $G(h+nQ) \leq \frac{1}{2} - \varepsilon$.
Then, thanks to the density argument, there are arbitrarily large values of $i$ for which $G(h+n_iQ) \leq \frac{1}{2} - \frac{\varepsilon}{2}$.
Since for $i$ sufficiently large, $|F(h+n_iQ)-G(h+n_iQ)| \leq \frac{\varepsilon}{2}$ (using (\ref{eq:F-G=0})), one has $F(h+n_iQ) \leq \frac{1}{2}$, and so $a^{h+n_iQ}\notin L$, contradicting condition \ref{cond:ud} of the statement.

\medskip

Therefore, $G(h+nQ) \geq \frac{1}{2}$ for large enough $n$. Because $G$ is not identically equal to $\frac{1}{2}$ (if it is the case, $F$ would be as close to $\frac{1}{2}$ as one wants, which is impossible since $L\in\BAfL$), again using density, there must be some $\varepsilon > 0$ and $k_0$ such that $G(h+k_0Q) \geq \frac{1}{2}+\epsilon$. 

First if $A = \emptyset$, it means that all the angles of the eigenvalues $\theta_1,  \dots, \theta_{s'}$ are rational. We can then write them as $\theta_k = \frac{l_k}{m_k}$. Then $G(n)$ takes a finite number of values, and these values only depend on $(n\mod m_1), \dots, (n \mod m_{s'})$.
Let's call $k_1 = h+k_0Q$ the number where $G$ is larger than $\frac{1}{2}$: $G(n_1) > \frac{1}{2}$. 
$G$ has the same value for all $n\in Z=\{ k_1 + k m_1\cdots m_{s'} |k\in\N \}$ (because for $n$ in this set, the values of all $(n\mod m_1), \dots, (n \mod m_{s'})$ are the same). 
Then, thanks to (\ref{eq:F-G=0}), one has, for $n\in Z$ sufficiently large, $F(n)>\frac{1}{2}$, so $\{ a^n ~|~ n\in Z, n\geq n_1 \} \subseteq L$. 
And because 
$
  \left|\{ a^n ~|~ n\in Z, n\geq n_1 \}\right|
  \sim \frac{n}{m_1 \cdots m_{s'}} 
$,
one has $\underline{dens}(L) > 0$, which contradicts condition \ref{cond:densL} of the statement.

Next, if $A\neq \emptyset$.
Let
\[
R((x_k)_{k\in A}) = \frac
                {\sum_{i\in E_a} \left|\sum_{k\in A} a_{ik} x_k + c \right|}
                {\sum_{i\in E}   \left|\sum_{k\in A} a_{ik} x_k + d \right|}   .
\]
Note that $G(h+n_i Q) = R((e^{2i\pi (h + n_i Q)\theta_k})_{k\in A})$. Then, because the sequences $((h+n_iQ)\theta_k)_i$ are uniformly distributed modulo 1, it follows that any value obtained by the function $R((e^{2i\pi y_k})_{k\in A})$ can be approximated by some $G(h+n_i Q)$ with arbitrary precision.
The function $R$ is continuous, therefore there exists an interval $I=(x_1, y_1, ...) = ((x_k, y_k))_{k\in A}$ on which $R((x_k)) > \frac{1}{2} + \frac{\varepsilon}{2}$. So, if $n_i$ is large enough and satisfies
\[
  \left( (h+n_iQ)\theta_1 \mod 1, \dots  \right) = \left( (h+n_iQ)\theta_k \mod 1 \right)_{k\in A} \in I  ,
\]
then $G(h+n_iQ) > \frac{1}{2} + \frac{\varepsilon}{2}$, which implies $F(h+n_iQ) > \frac{1}{2}$ and hence $a^{h+n_iQ}\in L$. Now we just have to prove that the sequence $(h+n_iQ)$ is ``dense enough'' to have $\underline{dens}(L) > 0$, contradicting again condition \ref{cond:densL}.\\
Because of uniform distribution imposed by condition \ref{cond:ud}, one has
\[
  d = \lim_{i\rightarrow\infty} \frac{C(I, h+n_iQ)}{h+n_iQ} = \prod_{k\in A} (y_k - x_k)
\]
And so for $i$ large enough, $\frac{C(I, h+n_iQ)}{h+n_iQ} \geq \frac{d}{2}$, with $a^{h+n_iQ}\in L$, implying $\underline{dens}(L) > 0$. We have proved that $L$ cannot be affine.
\qed
\end{proof}


Turakainen \cite{Turakainen1981} proved that  $ Prime = \{a^p \mid p \mbox{ is prime}\} $ and $ Poly(q) = \{ a^{q(n)} \mid n \in \mathbb{N}, q(n) \geq 0  \} $ (where $ q $ is any polynomial of degree $>2$ with non-negative coefficients) both satisfy the two conditions of Theorem \ref{th:non_affine}. Hence they are not in $ \BAfL $ .

\begin{coro}
  \label{coro:poly_non_affine}
  $ Prime \notin \BAfL $ and $ Poly(q) \notin \BAfL $.
\end{coro}

\section{Conclusion}
\label{sec:conc}

In this paper we demonstrated that even if they are strictly more powerful, bounded-error languages of affine automata share stability properties with regular languages (which are bounded-error languages of stochastic automata).

We also showed that the computational power of affine automata does not come alone from the unboundedness state vector set: the general model of unbounded state vector set can always be simulated with a bounded state vector set.
Hence some of the computational power of affine automata comes from the nonlinear nature of the final projection,  at least in the case of unbounded-error computation.

\bibliographystyle{splncs03}
\bibliography{affine}

\end{document}